\documentclass[12pt]{article}
\usepackage{fullpage}
\usepackage{graphicx}
\usepackage{amssymb}
\usepackage{amsmath, amsthm}
\usepackage{color}
\usepackage{epstopdf}
\usepackage{epsfig}
\newtheorem{theorem}{Theorem}[section]
\newtheorem{lemma}[theorem]{Lemma}

\newtheorem{corollary}[theorem]{Corollary}

\title{The
Complexity of Finding Multiple Solutions to Betweenness and 
Quartet Compatibility}
\author{Maria Luisa Bonet\thanks{
  Lenguajes y Sistemas Informaticos (LSI),
  Universitat Polit\`{e}cnica de Catalunya (UPC),
  08034 Barcelona, Spain,
  \tt{bonet@lsi.upc.edu}.}
\and Simone Linz\footnotemark[1] \thanks{Corresponding author: {\tt simone\_linz@yahoo.de}.}
\and Katherine St.~John\thanks{
  Department of Math \& Computer Science,
  Lehman College-- City University of New York (CUNY),
  Bronx, NY 12581, United States,
  \tt{ stjohn@lehman.cuny.edu}.}
  \footnote{
      Department of Computer Science,
  CUNY Graduate Center, New York, NY 10016.}
}

\begin{document}
\maketitle

\begin{abstract}
We show that two important problems that have applications in computational biology are ASP-complete, which implies that, given a solution to a problem, it is NP-complete
to decide if another solution exists.  We show first that a variation of {\sc Betweenness}, which is the underlying problem of questions related to radiation hybrid
mapping, is ASP-complete. Subsequently, we use that result to show that 
{\sc Quartet Compatibility}, a fundamental problem in phylogenetics that asks whether a set of quartets can be represented by a parent tree, is also ASP-complete.  The latter result
shows that Steel's {\sc Quartet Challenge}, which asks whether a solution to {\sc Quartet Compatibility} is unique, is coNP-complete. 
\end{abstract}

\section{Introduction}

Many biological problems focus on synthesizing data to 
yield new information.  We focus on the complexity of two
such problems.  The first is motivated by radiation hybrid
mapping (RH mapping) \cite{cox90} which  was developed to construct long range
maps of mammalian chromosomes (e.g.~\cite{Amaral2008a,Chowdhary2003a,Gyapay1996a}).  
Roughly, RH mapping uses x-rays to break the DNA into fragments and gives the relative order of DNA markers on
the fragments \cite{cox90}.  
The underlying computational problem is to assemble these 
fragments into a single strand (i.e., a ``linear order'').
As Chor and Sudan \cite{chor98} show, the
assembly of these fragments can be modeled by the well-known decision problem {\sc Betweenness}. Loosely speaking, this 
problem asks if there exists a total ordering over a set of elements that satisfies a set of constraints, each specifying one element to lie between two other elements (see Section~\ref{sec:prelim} for a more detailed definition).  
Since this
classical problem is NP-complete \cite{opatrny79},  Chor and
Sudan \cite{chor98} developed a polynomial-time approximation. Their algorithm, by using a geometric approach, either returns that no betweenness ordering exists or returns such an ordering that satisfies at least one half of the given constraints.  However, in the context of RH mappings, one usually knows the so-called 3' and 5' end of the DNA under consideration. Therefore, we consider a variation of {\sc Betweenness}---called {\sc cBetweenness}---whose instances do not only contain a collection of constraints, but also an explicit specification of the first and last DNA marker, and that asks whether or not there exists a total ordering whose first and last element coincide with the first and last DNA marker, respectively. In this paper, we are particularly interested in the following question related to {\sc cBetweenness}: given a solution to an instance of {\sc cBetweenness}, is there another solution? A positive answer to this question may imply
that the sampling of relatively ordered DNA markers is not large enough to determine
the correct ordering of the entire set of DNA markers.

Our second question focuses on finding the optimal phylogenetic 
tree for a set of taxa (e.g.~species).  Under the most popular optimization
criteria (maximum parsimony and maximum likelihood), it is NP-hard to find the optimal tree 
\cite{fouldsGraham,rochML}.  Despite the NP-hardness, several approaches to this problem exists, one of which finds a phylogenetic tree by
splitting the problem into subproblems, solves the 
subproblems and then recombines the solutions to form a complete phylogenetic tree that represents all taxa under consideration 
\cite{berry2000,buneman,dcm99}.  To this end, quartets which are phylogenetic trees
on 4 taxa are often used. Given 4 taxa, there are 
three possible ways to arrange them:

\vspace{.25in}
  \thicklines
  \setlength{\unitlength}{2.5pt}%
  \hbox to\textwidth{\hfil
    \begin{picture}(40,35)(0,-12)
      \put(0,0){\circle*{2}}
      \put(0,20){\circle*{2}}
      \put(40,0){\circle*{2}}
      \put(40,20){\circle*{2}}
      \put(10,10){\circle*{1}}
      \put(30,10){\circle*{1}}
      \put(0,0){\line(1,1){10}}
      \put(0,20){\line(1,-1){10}}
      \put(40,0){\line(-1,1){10}}
      \put(40,20){\line(-1,-1){10}}
      \put(10,10){\line(1,0){20}}
      \put(-4,0){\makebox(0,0){ $a$}}
      \put(-4,20){\makebox(0,0){$b$}}
      \put(44,0){\makebox(0,0){$c$}}
      \put(44,20){\makebox(0,0){$d$}}
      \put(20,-10){\makebox(0,0){$ab|cd$}}
    \end{picture}%
    \hfil
    \begin{picture}(40,35)(0,-12)
      \put(0,0){\circle*{2}}
      \put(0,20){\circle*{2}}
      \put(40,0){\circle*{2}}
      \put(40,20){\circle*{2}}
      \put(10,10){\circle*{1}}
      \put(30,10){\circle*{1}}
      \put(0,0){\line(1,1){10}}
      \put(0,20){\line(1,-1){10}}
      \put(40,0){\line(-1,1){10}}
      \put(40,20){\line(-1,-1){10}}
      \put(10,10){\line(1,0){20}}
      \put(-4,0){\makebox(0,0){$a$}}
      \put(-4,20){\makebox(0,0){$c$}}
      \put(44,0){\makebox(0,0){$b$}}
      \put(44,20){\makebox(0,0){$d$}}
      \put(20,-10){\makebox(0,0){$ac|bd$}}
    \end{picture}%
    \hfil
    \begin{picture}(40,35)(0,-12)
      \put(0,0){\circle*{2}}
      \put(0,20){\circle*{2}}
      \put(40,0){\circle*{2}}
      \put(40,20){\circle*{2}}
      \put(10,10){\circle*{1}}
      \put(30,10){\circle*{1}}
      \put(0,0){\line(1,1){10}}
      \put(0,20){\line(1,-1){10}}
      \put(40,0){\line(-1,1){10}}
      \put(40,20){\line(-1,-1){10}}
      \put(10,10){\line(1,0){20}}
      \put(-4,0){\makebox(0,0){$a$}}
      \put(-4,20){\makebox(0,0){$d$}}
      \put(44,0){\makebox(0,0){$b$}}
      \put(44,20){\makebox(0,0){$c$}}
      \put(20,-10){\makebox(0,0){$ad|bc$}}
    \end{picture}%
    \hfil}\smallskip
\noindent Since the number of possible topologies grows exponentially with the number of taxa, it is easier to decide the arrangement or
topology of each subset of 4 taxa than the topology for a large set of taxa.  This leads to the question:
how hard is it to build a tree from quartets?  If $Q$ denotes the set of all quartets
of a phylogenetic tree $T$, then $T$ is uniquely determined by $Q$ and can be reconstructed in polynomial time \cite{erdos}.
However, in most cases $T$ is not given, and $Q$ is often incomplete (i.e.~there exists a set of 4 taxa for which no quartet is given) or elements in $Q$ contradict each other. In such a case, it is NP-complete to decide whether there exists
a phylogenetic tree on $n$ taxa that {\it displays} $Q$ \cite{steel92}; that is, a phylogenetic tree that explains all the ancestral relationships given by the quartets, where $n$ is the number of taxa over all elements in $Q$. Because of the NP-completeness of this latter problem---called {\sc Quartet Compatibility}---algorithmic approaches are rare. Nevertheless, several attractive graph-theoretic characterizations of the problem exist~\cite{grunewald08,semple02}. While~\cite{semple02} approaches the problem by using a chordal-graph characterization on an underlying intersection graph,~\cite{grunewald08} establishes a so-called quartet graph and edge colorings via this graph to decide whether or not there exists a phylogenetic tree that displays a given set of quartets. As a follow-up on this last question, Steel asks whether or not the following problem, called the {\sc Quartet Challenge}, is NP-hard~\cite{quartetChallenge}:
given a set $Q$ of quartets over $n$ taxa and a phylogenetic tree $T$ on $n$ taxa that displays $Q$, is $T$ the unique such tree that displays $Q$? Although the above-mentioned characterizations~\cite{grunewald08,semple02} comprises several results on when a set of quartets is displayed by a unique phylogenetic tree, the computational complexity of the {\sc Quartet Challenge} remains open.
We note that if a set $Q$ of quartets that does not contain any redundant information is displayed by a unique phylogenetic tree on $n$ taxa, then the minimum size of $Q$ is $n-3$ \cite[Corollary 6.3.10]{sempleSteelBook} while the current largest maximum size is $2n-8$~\cite{dowden10}.

To investigate the computational complexity of problems for when a solution is given and one is interested in finding another solution, Yato and Seta~\cite{yato03} developed the framework of {\sc Another Solution Problems (ASP)}. Briefly, if a problem is ASP-complete, then given a solution to a problem, it is NP-complete to decide if a distinct solution exists. Many canonical problems are ASP-complete, such as several variations
of satisfiability, as well as games like Sudoku \cite{yato03}. 

In this paper, we show that, given a solution to an instance of {\sc cBetweenness} or {\sc Quartet Compatibility}, finding a second solution is ASP-complete. 
To show that {\sc cBetweenness} is ASP-complete,
we use a reduction from a variant of satisfiability, namely,
{\sc Not-All-Equal-3SAT with Constants}
(see Section~\ref{sec:prelim} for a detailed definition). Using this result, we establish a second reduction that subsequently shows that  {\sc Quartet Compatibility} is also ASP-complete. As we will soon see, the  ASP-completeness of {\sc Quartet Compatibility} implies coNP-completeness of the {\sc Quartet Challenge}.


We note that while this paper was in preparation, a preprint was released
\cite{Habib2010a} that addresses the complexity of the {\sc Quartet Challenge} using different
techniques than employed here.

This paper is organized as follows:  Section~\ref{sec:prelim} details background 
information from complexity theory and phylogenetics.  Section~\ref{sec:results} gives 
the two reduction results, and Section~\ref{sec:discussion} contains some concluding
remarks.

\section{Preliminaries}\label{sec:prelim}
This section gives an outline of the ASP-completeness concept and formally states the decision problems that are needed for this paper. Preliminaries in the context of phylogenetics are given in the second part of this section.

\subsection{Computational complexity}
Notation and terminology introduced in this section follows \cite{gareyJohnson} and \cite{yato03}, with the former being an excellent reference for 
general complexity results.\\

{\bf ASP-completeness.}
The notion of {\it ASP-completeness} was first  published by Yato and Seta~\cite{yato03}. Their paper provides a theoretical framework to analyze the computational complexity of problems whose input contains, among others, a solution to a given problem instance, and the objective is to find a distinct solution to that instance or to return that no such solution exists. To this end, the authors use {\it function problems} whose answers can be more complex in contrast to decision problems that are always answered with either {\it `yes'} or {\it `no'}. Formally, the complexity class {\it FNP} contains each function problem $\Pi$ that satisfies the following two conditions:
\begin{itemize}
\item[(i)] There exists a polynomial $p$ such that the size of each solution to a given instance $\psi$ of $\Pi$ is bounded by $p(\psi)$.
\item[(ii)] Given an instance $\psi$ of $\Pi$ and a solution $s$, it can be decided in polynomial time if $s$ is a solution to $\psi$.
\end{itemize}

\noindent Note that the complexity class FNP is a generalization of the class NP and that each function problem in FNP has an analogous decision problem. 

Now, let $\Pi$ and $\Pi'$ be two function problems. We say that a polynomial-time reduction $f$ from $\Pi$ to $\Pi'$ is an {\it ASP-reduction} if for any instance $\psi$ of $\Pi$ there is a bijection from the solutions of $\psi$ to the solutions of $f(\psi)$, where $f(\psi)$ is an instance of $\Pi'$ that has been obtained under $f$. Note that, while each ASP-reduction is a so-called `parsimonious reduction', the converse is not necessarily true. Parsimonious reductions have been introduced in the context of enumeration problems (for details, see~\cite{papadimitriou03}). 
Furthermore, a function problem $\Pi'$ is {\it ASP-complete} if and only if $\Pi'\in \textnormal{FNP}$ and there is an ASP-reduction from $\Pi$ to $\Pi'$ for any function problem $\Pi\in \textnormal{FNP}$. \\

\noindent{\bf Remark.} Throughout this paper, we prove that several problems are ASP-complete. Although these problems are stated as decision problems in the remainder of this section, it should be clear from the context which associated function problems we consider. More precisely, for a decision problem $\Pi_d$, we consider the function problem $\Pi$ whose instance $\psi$ consists of the same parameters as an instance of $\Pi_d$ and, additionally, of a solution to $\psi$, and whose question is to find a distinct solution that fulfills all conditions given in the question of $\Pi_d$. Hence, if $\Pi$ is ASP-complete, then this implies that, unless P=NP, it is computationally hard to find a second solution to an instance of $\Pi$.\\

{\bf Satisfiability (SAT) problems.} The satisfiability problem is a well-known problem in the study of computational complexity. In fact, it was the first problem shown to be NP-complete \cite{gareyJohnson,cook}. Before we can formally state the problem, we need some definitions. 
Let $V = \{x_1,x_2,\ldots,x_n\}$ be a set of variables. A {\it literal} is either a variable $x_i$ or its negation $\bar x_i$, and a {\it clause} is a disjunction of literals. Now let $C$ be a conjunction of clauses (for an example, consider the four clauses given in Figure~\ref{fig:ex1}). A {\it truth assignment} for $C$ assigns each literal to either {\it true} or {\it false} such that, for each $i\in\{1,2,\ldots,n\}$, $x_i=true$ if and only if $\bar x_i=false$. We say that a literal is {\it satisfied} (resp. {\it falsified}) if its truth value is {\it true} (resp. {\it false}). 

\begin{quote}
{\bf Problem:} {\sc Satisfiability}\\
{\bf Instance:} A set of variables $V$ and a conjunction $C$ of clauses over $V$.\\
{\bf Question:} Does there exist a truth assignment for $C$ such that each clause contains at least one literal assigned to {\it true}?
\end{quote}

\begin{figure}
\begin{center}
\begin{tabular}{lll}
$C_1: x_1 \vee x_2 \vee x_3 $ &\mbox{\hspace{.5in}}& $\sigma_0: \{x_1,x_2,x_3,x_4\} \rightarrow true$ \\
$C_2: x_1 \vee \bar x_3 \vee x_4$ && $\sigma_1: \{x_1,x_2\}\rightarrow true; \sigma_1: \{x_3,x_4\} \rightarrow false$\\
$C_3: \bar x_1 \vee x_3 \vee \bar x_4$ && $\sigma_2: \{x_1,x_3,x_4\} \rightarrow true; \sigma_2: \{x_2\} \rightarrow false$\\
$C_4: x_2 \vee \bar x_3 \vee x_4$ & \\

\end{tabular}
\end{center}

\caption{\small Left:  An example of 4 clauses on the variables $x_1,\ldots,x_4$.  Right: Truth assignments to the non-negated literals.
As an instance of {\sc 3SAT}, there are several possible satisfying truth assignments including $\sigma_0$, $\sigma_1$, and $\sigma_2$.
When viewed as an instance of {\sc NAE-3SAT}, $\sigma_0$ is not satisfying since
it assigns all literals of the first clause $C_1$ to {\it true}, violating the `not all equal' condition.}
\label{fig:ex1}
\end{figure}


{\sc 3SAT} is a special case of the general SAT problem in which each clause of a given instance contains exactly three literals. Referring back to Figure~\ref{fig:ex1}, all three truth assignments $\sigma_0$, $\sigma_1$, and $\sigma_2$ satisfy the four clauses for when regarded as an instance of {\sc 3SAT}. The next theorem is due to~\cite[Theorem 3.5]{yato03}.

\begin{theorem}
{\sc 3SAT} is ASP-complete.
\end{theorem}

We will next show the ASP-completeness of another version of SAT that is similar to the following decision problem.

\begin{quote}
{\bf Problem:} {\sc Not-All Equal-3SAT (NAE-3SAT)}\\
{\bf Instance:} A set of variables $V$ and a conjunction $C$ of 3-literal clauses over $V$.\\
{\bf Question:} Is there a truth assignment such that for each clause there is a literal satisfied and a literal falsified by the assignment?
\end{quote}
As an example, see Figure~\ref{fig:ex1}, and note that $\sigma_0$ does not satisfy the four clauses when regarded as an instance of {\sc NAE-3SAT}. It is an immediate consequence of the definition of NAE-3SAT that, given a solution $S$ to an instance, a second solution to this instance can be calculated in polynomial time by taking the complement of $S$; that is assigning each literal to {\it true} (resp. {\it false}) if it is assigned to {\it false} (resp. {\it true}) in $S$ (see~\cite{juban99}). 

The next decision problem can be obtained from NAE-3SAT by allowing for instances that contain the constants $T$ or $F$. 
\begin{quote}
{\bf Problem:} {\sc Not-All Equal-3SAT with constants} ({\sc cNAE-3SAT})\\
{\bf Instance:} A set of variables $V$, constants $T$ and $F$, and a conjunction $C$ of 3-literal clauses over $V\cup\{T,F\}$.\\
{\bf Question:} Is there a truth assignment such that the constants $T$ and $F$ are assigned {\it true} and {\it false}, respectively, and 
for each clause, there is a literal or constant satisfied and a literal or constant falsified by the assignment?

\end{quote}

\noindent In the case of {\sc cNAE-3SAT}, we cannot always obtain a second solution from a first one by
taking its complement. For instance, if an instance of {\sc cNAE-3SAT} contains the clause $a_k\vee b_k\vee T$, then the assignment $a_k=b_k=false$ is valid, while the complementary assignment $a_k=b_k=true$ is not. In fact, we next show that {\sc cNAE-3SAT} is ASP-complete.

\begin{theorem}
{\upshape {\sc cNAE-3SAT}} is ASP-complete.
\end{theorem}

\begin{proof}
Regarding {\sc cNAE-3SAT} as a function problem (see the remark earlier in this section), it is easily seen that deciding if a truth assignment to an instance of {\sc cNAE-3SAT} satisfies this instance  takes polynomial time. Hence, {\sc cNAE-3SAT} is in FNP. Now, let $\psi$ be an instance of the APS-complete problem 3SAT over the variables $V=\{x_1,\ldots,x_n\}$. To show that {\sc cNAE-3SAT} is ASP-complete we reduce $\psi$ to an instance $\psi'$ of {\sc cNAE-3SAT} over an expanded set of variables $V' = V \cup \{x_{x+1},\ldots,x_{n+m}\}$, where $m$ is the number of clauses in $\psi$. Let $x_{n+k}$ be a new variable chosen for the clause $(a_k\vee b_k\vee c_k)$ of $\psi$. We obtain $\psi'$ from $\psi$ by replacing each clause $(a_k\vee b_k\vee c_k)$ with the following 4 clauses:
$$(a_k\vee b_k\vee x_{n+k})\wedge(\bar x_{n+k}\vee c_k\vee F)\wedge(a_k\vee x_{n+k}\vee T)\wedge(b_k\vee x_{n+k}\vee T).$$ 
Clearly, this reduction can be done in polynomial time. The size of $\psi'$ is polynomial in the size of $\psi$, and a straightforward check shows that $\psi$ is satisfiable if and only if $\psi'$ is satisfiable. Furthermore, for each clause, $x_{n+k}$ is uniquely determined by the truth values of $a_k$ and $b_k$, since the
reduction makes $\bar x_{n+k}$ equivalent to $a_k\vee b_k$. Hence, 
each truth assignment that satisfies $\psi$ can be mapped to a unique valid truth assignment of $\psi'$. Consequently, the converse; i.e. each truth assignment of $\psi'$ is mapped to a unique truth assignment of $\psi$, also holds. It now follows that the described reduction from 3SAT to {\sc cNAE-3SAT} is an ASP-reduction, thereby completing the proof of this theorem.
\end{proof}


{\bf The {\sc Betweenness} problem.} The decision problem {\sc Betweenness}, that we introduce next, asks whether or not a given finite set can be totally ordered such that a collection of constraints which are given in form of triples is satisfied.

\begin{quote}
{\bf Problem:} {\sc Betweenness}\\
{\bf Instance:} A finite set $A$ and a collection $C$ of ordered triples $(a,b,c)$ of distinct elements from $A$ such that each element of $A$ occurs in at least one triple from $C$. \\
{\bf Question:} Does there exist a {\it betweenness ordering} $f$ of $A$ for $C$; that is a one-to-one function $f:A\rightarrow\{1,2,\ldots,|A|\}$ such that for each triple $(a,b,c)$ in $C$ either $f(a)<f(b)<f(c)$ or $f(c)<f(b)<f(c)$?
\end{quote}
\noindent Loosely speaking, for each triple $(a,b,c)$, the element $b$ lies between $a$ and $c$ in a betweenness ordering of $A$ for $C$.  {\sc Betweenness} has been shown to be NP-complete~\cite{opatrny79}. Similar to NAE-3SAT, notice that if there is a solutions, say $a_1<a_2<\ldots<a_s$, to an instance of {\sc Betweenness}, then there is  always a second solution $a_s<\ldots<a_2<a_1$ to that instance that can clearly be calculated in polynomial time. Therefore, {\sc Betweenness} is not ASP-complete. \\

We next introduce a natural variant of {\sc Betweenness}---called {\sc cBetweenness}---that is ASP-complete (see Section~\ref{sec:Betweenness}). An instance of {\sc cBetweenness} differs from an instance of {\sc Betweenness} in a way that the former contains two constants, say $m$ and $M$, and each betweenness ordering has $m$ as its first and $M$ as its last element. We say that $m$ is the {\it minimum} and $M$ the {\it maximum} of each betweenness ordering.
\begin{quote}
{\bf Problem:} {\sc cBetweenness}\\
{\bf Instance:} A finite set $A$ and a collection $C$ of ordered triples $(a,b,c)$ of distinct elements from $A\cup\{M,m\}$ with $m,M\notin A$ such that each element of $A\cup\{M,m\}$ occurs in at least one triple from $C$. \\
{\bf Question:} Does there exist a {\it betweenness ordering} $f$ of $A\cup\{M,m\}$ for $C$ such that $f$ is a one-to-one function $f:A\cup\{m,M\}\rightarrow\{0,1,2,\ldots,|A|+1\}$ such that for each triple $(a,b,c)$ in $C$ either $f(a)<f(b)<f(c)$ or $f(c)<f(b)<f(c)$, and $f(m)=0$ and $f(M)=|A|+1$?
\end{quote}
Although an instance $\psi$ of {\sc cBetweenness} can have several betweenness orderings, note that if $a_0<a_1<a_2<\ldots<a_{|A|̣+1}$ is a betweenness ordering for $\psi$ with $a_0=m$ and $a_{|A|+1}=M$, then $a_{|A|+1}<\ldots<a_2<a_1<a_0$ is not such an ordering. This is because $m$ must be the minimal element, and $M$ the maximal. 

\subsection{Phylogenetics}

This section provides preliminaries in the context of phylogenetics. For a more thorough overview, we refer the interested reader to Semple and Steel \cite{sempleSteelBook}.\\

{\bf Phylogenetic trees and subtrees.} An {\it unrooted phylogenetic $X$-tree} $T$ is a connected acyclic graph whose leaves are bijectively labeled with elements of $X$ and have degree 1. Furthermore, $T$ is {\it binary} if each non-leaf vertex has degree 3. The set $X$ is the {\it label set} of $T$ and denoted by $L(T)$. 

Now let $T$ be an unrooted phylogenetic $X$-tree, and let $X'$ be a subset of $X$. The {\it minimal subtree} of $T$ that connects all elements in $X'$ is denoted by $T(X')$. Furthermore, the {\it restriction} of $T$ to $X'$, denoted by $T|X'$, is the phylogenetic tree obtained from  $T(X')$ by contracting degree-2 vertices.

Throughout this paper, we will use the terms `unrooted phylogenetic tree' and `phylogenetic tree' interchangeably.\\

{\bf Quartets.} A {\it quartet} is an unrooted binary phylogenetic tree with exactly four leaves. For example, let $q$ be a quartet whose label set is $\{a,b,c,d\}$. We write $ab|cd$ (or equivalently, $cd|ab$) if the path from $a$ to $b$ does not intersect the path from $c$ to $d$. Similarly to the label set of a phylogenetic tree, $L(q)$ denotes the label set of $q$, which is $\{a,b,c,d\}$. Now, let $Q=\{q_1,q_2,\ldots,q_n\}$ be a set of quartets. We write $L(Q)$ to denote the union $L(q_1)\cup L(q_2)\cup\ldots\cup L(q_n)$.\\

{\bf Compatibility.} Let $T$ be a phylogenetic tree whose leaf set is a superset of $L(q)$. Then $T$ {\it displays} $q$ if $q$ is isomorphic to $T|L(q)$. Furthermore, $T$ displays a set $Q$ of quartets if $T$ displays each element of $Q$, in which case $Q$ is said to be {\it compatible}. Lastly, ${<}Q{>}$ denotes the set of all unrooted binary phylogenetic trees that display $Q$ and whose label set is precisely $L(Q)$. 

The concept of compatibility leads to the following decision problem, which has been shown to be NP-complete~\cite{steel92}.

\begin{quote}
{\bf Problem:} {\sc Quartet Compatibility}\\
{\bf Instance:} A set $Q$ of quartets.\\
{\bf Question:} Is $Q$ compatible?
\end{quote}

The next problem has originally been posed by Steel~\cite{quartetChallenge} and is a natural extension of {\sc Quartet Compatibility}.

\begin{quote}
{\bf Problem:} {\sc Quartet Challenge}\\
{\bf Instance:} A binary phylogenetic $X$-tree $T$ and a set $Q$ of quartets on $X$ such that $T$ displays $Q$.\\
{\bf Question:} Is $T$ the unique phylogenetic $X$-tree that displays $Q$?
\end{quote}

\noindent{\bf Remark.} Given a binary phylogenetic $X$-tree $T$ and a set $Q$ of quartets on $X$ such that $T$ displays $Q$, deciding whether another solution exists is the complement question of the {\sc Quartet Challenge}. That is, a {\it no} answer to an instance of the first question translates to a {\it yes} answer to the same instance of the {\sc Quartet Challenge} and vice versa.\\

We end this section by highlighting the relationship between the two problems {\sc Quartet Challenge} and {\sc Quartet Compatibility}.  Let $Q$ be an instance of {\sc Quartet Compatibility}. If $Q$ is compatible, then there exists an unrooted phylogenetic tree $T$ with label set $L(Q)$ that displays $Q$. This naturally leads to the question whether $T$ is the unique such tree on $L(Q)$, which is exactly the question of the {\sc Quartet Challenge}. To make progress towards resolving the complexity of this challenge, we will first show that {\sc Quartet Compatibility} is ASP-complete. Thus, by ~\cite[Theorem 3.4]{yato03}, the decision problem that corresponds to {\sc Quartet Compatibility}, i.e. asking whether another solution exists, is NP-complete. Now, recalling the last remark, this in turn implies that the {\sc Quartet Challenge} is coNP-complete.
Lastly, note that if $T$ is the unique tree with label set $L(Q)$ that displays $Q$, then $T$ is binary.

\section{ASP-Completeness Results}
\label{sec:results}

\subsection{ASP-Reduction for {\sc cBetweenness}}\label{sec:Betweenness}

In this section, we focus on establishing an ASP-reduction from {\sc cNAE-3SAT} 
to {\sc cBetweenness}.  This reduction gives that 
{\sc cBetweenness} is ASP-complete (see Theorem~\ref{t:cBetweenness}). We note that, by a similar argument, we
can also reduce {\sc NAE-3SAT} to {\sc Betweenness}.

Let $\psi$ be an instance of {\sc cNAE-3SAT} 
consisting of a conjunction of 3-literal clauses
$\{a_k \vee b_k \vee c_k: 1\le k\le l\}$,
on the set of variables $\{x_1,\ldots,x_n\}$ and a subset of the constants $\{T,F\}$. 
Let $\{-x_1,-x_2,\ldots,-x_n,x_1,x_2,\ldots,x_n\}$ be the set of corresponding literals of $\psi$, where $-x_i$ (instead of $\bar x_i$) denotes the negation of $x_i$. We next build an instance $\psi'$ of {\sc cBetweenness}. 

In the following, we think of an ordering as
setting symbols on a line segment where the symbol $X$ is
the center; i.e. $X$ represents `$0$'. Furthermore, we refer to the region left of $X$ as the {\it negative side} and to the region right of $X$ as the {\it positive side} of the line segment.
For every variable $x_i$ with $i\in\{1,2,\ldots,n\}$, we preserve the two symbols $x_i$ and
$-x_i$ and introduce two new symbols $M_i$ and $-M_i$ which are auxiliary symbols that mark the
midpoints between the $x_i$ and $-x_i$ symbols (see Figure~\ref{fig:order}). Moreover, if $T$ or $F$ is contained in any clause of $\psi$, then we also introduce the symbol $M$ and $m$, respectively, such that $M$ represents the largest and $m$ the smallest value on the line segment (for details, see below).
Intuitively, if the literal $x_i$ is assigned to {\it true}, then the symbol $x_i$ is assigned to the positive side and the symbol $-x_i$ to the negative side. Otherwise, if the literal $x_i$ is assigned to {\it false}, then the symbol $x_i$ is assigned to the negative side and the symbol $-x_i$ to the positive side.

The following triples fix $X$ as `$0$'.  For every $i\le n$, 
they put $x_i$ and $-x_i$ on either side of $X$. They also put 
$M_i$ and $-M_i$ on either side of $X$. 

\begin{equation}\label{eq1}
\begin{array}{ll}
(-x_i, X, x_i) &\mbox{ for all $i$ such that $1\le i\le n$}  \\
(-M_i, X, M_i) &\mbox{ for all $i$ such that $1\le i\le n$}
\end{array}
\end{equation}
The next set of triples put an order between the $x_i$ and $M_i$ symbols. They establish that
$|x_i|<|x_{i+1}|$ for every $i, 1\le i <n$, where we interpret $|\cdot |$ to be the distance
to $X$ (i.e.~`$0$') under the induced ordering.
Also, they fix $M_i$ and $-M_i$ as middle points
between $x_i$ (or $-x_i$), and $x_{i-1}$ (or $-x_{i-1}$) (see Figure~\ref{fig:order}).

\begin{equation}\label{eq2}
\begin{array}{lll}
(-M_{i}, x_{i-1}, M_{i})& (-M_{i}, -x_{i-1}, M_{i}) &\mbox{ for all $i$ such that $2\le i\le n$}\\
(-x_i, M_i, x_i)& (-x_i, -M_i, x_i) &\mbox{ for all $i$ such that $1\le i\le n$} \\
\end{array}
\end{equation} 
We will require that either both $x_i$ and $M_i$ are on the positive side, or on the negative side.
\begin{equation}\label{eq3}
\begin{array}{lll}
(x_i, X, -M_i)&  (-x_i, X, M_i)& \mbox{ for all $i$ such that $1\le i\le n$}\\
\end{array}
\end{equation}

\begin{figure}
\begin{center}
\includegraphics[height=1.5in]{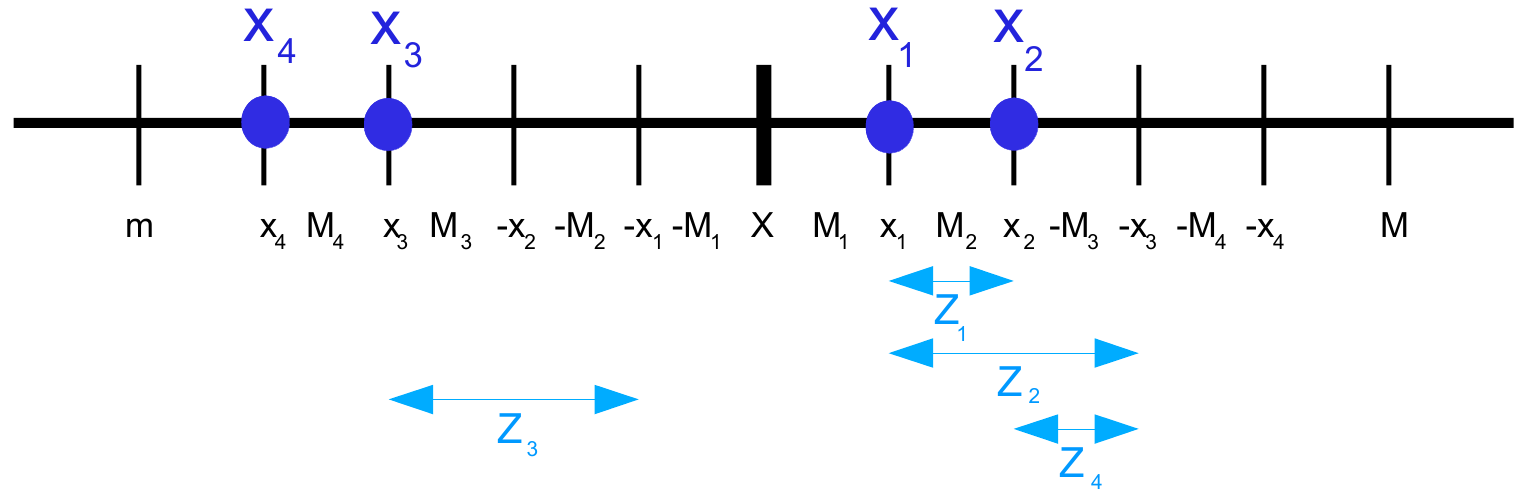}
\end{center}
\caption{\small A mapping of the first set of clauses under the truth assignment $\sigma_1$
from Figure~\ref{fig:ex1}.  The top blue variables are the initial variables from
the instance of {\sc cNAE-3SAT}.  For a given truth assignment, the ordering of the 
$x_i$ and $M_i$ symbols is fixed by the triples in Equation~\ref{eq1}-\ref{eq4}.  However, 
without additional triples (and  auxiliary symbols), the $Z_i$ symbols have overlapping
ranges (as indicated by the lines with arrows) and, as such, several possible orderings
for a single truth assignment.}
\label{fig:order}
\end{figure}

\paragraph{\bf Original Encoding:}
Given a clause $(a_k \vee b_k \vee c_k)$ of $\psi$, we assume for the remainder of Section~\ref{sec:Betweenness} that $a_k\in\{-x_i,x_i\}$, $b_k\in\{-x_{i'},x_{i'}\}$, and $c_k\in\{-x_{i''},x_{i''}\}$ such that $i<i'<i''$. Since there are at most two constants per clause, we also assume that $a_k$ is never a constant and that, if $b_k$ is a constant, then $c_k$ is a constant. Lastly, we assume that the whole set of clauses of $\psi$ is ordered lexicographically.

To guarantee the `not-all-equal' condition, we use the original 
encoding of Opatrny \cite{opatrny79}. For each clause $a_k\vee b_k\vee c_k$ with $k\in\{1,2,\ldots,l\}$,
we add a new symbol,
$Z_k$, where $l$ is the number of clauses in $\psi$.
We add the following triples that correspond to the original encoding of
Opatrny \cite{opatrny79}:  
\begin{equation}\label{eq4}
\begin{array}{lll}
(a_k, Z_k, b_k) & (c_k, X, Z_k) &\mbox{ for all $k$ such that $1\le k\le l$.}
\end{array}
\end{equation}
If a clause contains a constant $F$ or $T$, it gets substituted by 
$m$ or $M$, respectively, in the triple. For instance, a clause $a_k \vee b_k \vee T$ generates the two triples $(a_k, Z_k, b_k)$, $(T, X, Z_k)$. These triples
force $Z_k$ to be on the negative side.

The triples of Equation \ref{eq4} say that at least one of the literals of the $k^{th}$ clause is assigned to
{\it true} and one to {\it false} (see Figure~\ref{fig:order}). 
When viewed in terms of the possible truth assignment to the variables
in the original clause, the truth assignment $a_k=b_k=c_k=true$
 is eliminated since for all three initial literals to 
be assigned {\it true}, both $c_k$ and $Z_k$ would be assigned to the 
positive side of $X$ violating the second triple of Equation~\ref{eq4}.  By a similar argument, the truth assignment $a_k=b_k=c_k=false$
is eliminated.  This leaves six 
other possible satisfying truth assignments.

The triples of Equation~\ref{eq1}-\ref{eq4}
are not sufficient to prove ASP-completeness of {\sc cBetweenness} since different orderings of {\sc cBetweenness} can be 
reduced to the same satisfying truth assignment of the {\sc cNAE-3SAT} instance
(see Figure~\ref{fig:order}).
This happens since the  symbols in $\{Z_1,Z_2,\ldots,Z_l\}$ are not fixed with respect
to the $x_i$ and $-x_i$ symbols or to one another. In what follows, we will 
show how to fix the ordering of these symbols.

\paragraph{\bf Fixing Auxiliary Symbols in the Ordering of the Initial Symbols:}
We next introduce new symbols $-Z_1,\ldots,-Z_l$, where $l$ is the number of clauses in $\psi$.
From now on, we will establish a set of triples for every clause. 
We assume that the triples up to the $(k-1)^{th}$ clause have been defined, and now we
will give the set of triples of the $k^{th}$ clause. The intuition is that 
we want $-Z_k$ to be opposite to $Z_k$ with respect to $X$. We add:
\begin{equation}\label{eq5}
\begin{array}{ll}
(-a_k, -Z_k, -b_k) &  (-c_k, X, -Z_k) \\
\end{array}
\end{equation}

To fix the position of $Z_k$ and $-Z_k$ for every $1\le k\le l$ relative to the $x_i$ and $M_i$ symbols, we need more
triples. Recall that $a_k\in\{-x_i,x_i\}$. We
add the following triples saying that $|Z_k|<|M_{i+1}|$.
\begin{equation}\label{eq6}
\begin{array}{ll}
(-M_{i+1}, Z_k, M_{i+1})& (-M_{i+1}, -Z_k, M_{i+1}) \\
\end{array}
\end{equation}                                                                                                
The following triples say that $|M_i|<|Z_k|$.
\begin{equation}\label{eq7}
\begin{array}{ll}
(-Z_k, M_i, Z_k)& (-Z_k, -M_i, Z_k) 
\end{array}
\end{equation}

These triples eliminate the interval $[-M_i,M_i]$ for the positions of $Z_k$ and $-Z_k$, respectively.
Since the $M_i$ symbols occur as midpoints between the $x_i$ symbols, we have restricted each $Z_k$ to be `near'
$a_k$ or $-a_k$; i.e. $Z_k$ is `near' $-x_i$ or $x_i$ on the line segment. As a consequence of Equations~\ref{eq6}, we have
$
|M_i| < |Z_k| < |x_i| \mbox{ or } |x_i| < |Z_k| < |M_{i+1}|.
$
At this point, we have the symbols $Z_k$ and $-Z_k$ in a tight interval between consecutive
positions, except for the truth assignments $a_k=c_k=false$ and $b_k=true$, and $a_k=c_k=true$ and $b_k=false$. In these cases, both
$|M_i| < |Z_k| < |x_i| \mbox{ and } |x_i| < |Z_k| < |M_{i+1}|$, are still possible. See Figure~\ref{fig:orderZtriples}
for an example. We therefore add:
\begin{equation}\label{eq8}
\begin{array}{ll}
(Z_k, -a_k, c_k) & (-Z_k, a_k, -c_k) 
\end{array}
\end{equation}
which fixes the $Z_k$ such that $|x_i| < |Z_k| < |M_{i+1}|$.

\begin{figure}
\begin{center}
\includegraphics[height=1.25in]{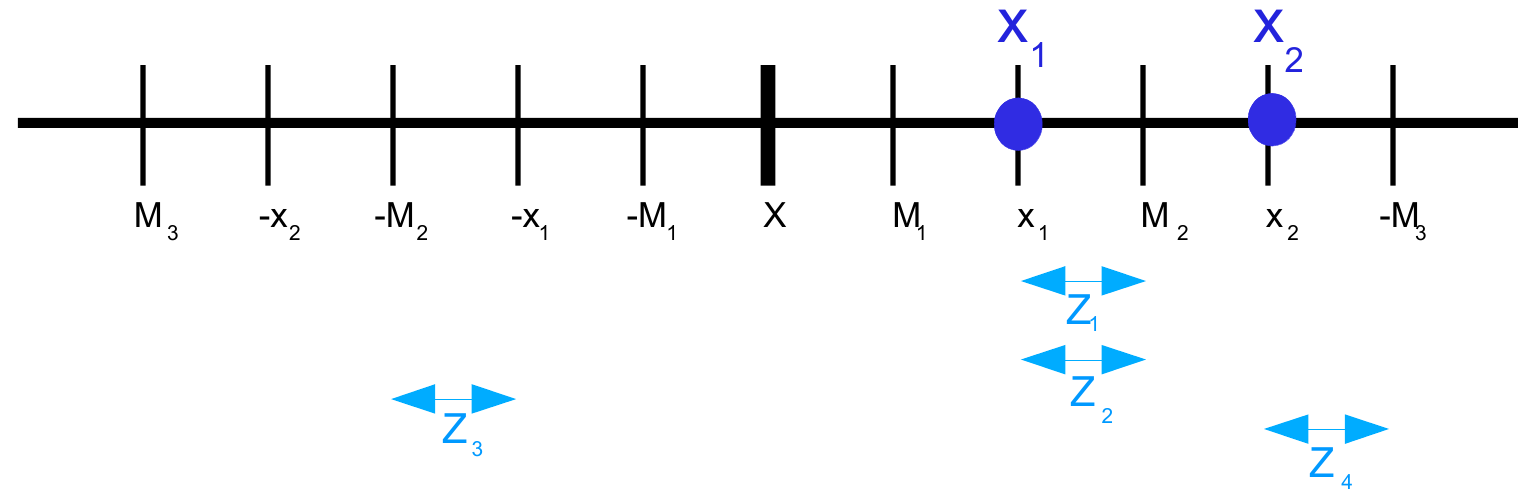}
\end{center}
\caption{\small Equations~\ref{eq6}-\ref{eq7} limit the ranges of the $Z_k$ symbols but
do not completely fix the ordering.  Continuing the example from Figure~\ref{fig:ex1}, we show
the possible ordering for the $\sigma_1$ truth assignment.  Note that the order of
$Z_1$ and $Z_2$ is not fixed and additional triples (and variables) are needed.}
\label{fig:orderZtriples}
\end{figure}

\paragraph{\bf Fixing the Order of Auxiliary Symbols Among Themselves:}
With the $Z_k$ and $-Z_k$ being fixed around $x_i$ and $-x_i$, respectively, and recalling that $a_i\in\{-x_i,x_i\}$, we need to address the
case where the first literal of several clauses is either $x_i$ or $-x_i$. Note that these corresponds to consecutive clauses of $\psi$ since $\psi$ is ordered lexicographically.  In this case, we will use extra auxiliary symbols
to nest the $Z_k$ symbols around $x_i$ and $-x_i$. The intuition for these
clauses is to partition the consecutive intervals 
into distinct regions for the different $Z_k$ and $-Z_k$ symbols (see Figure~\ref{fig:orderZ}).

Let the $k^{th}$ clause be $a_k\vee b_k \vee c_k$ and assume that it is not the only 
clause starting by $x_i$ or $-x_i$. We will add the symbols $L_k$, $-L_k$, $U_k$, $-U_k$ to create
4 new subintervals (two around $x_i$ and two around $-x_i$) where we will fix $Z_k$ and 
$-Z_k$. Also, we will add the new symbols
$Y_k$ to be the mirror image of $Z_k$ with respect to $x_i$ (or $-x_i$), and
$-Y_k$ to be the mirror image of $-Z_k$ with respect to $x_i$ (or $-x_i$).

The following triples express that the symbols $Z_k$, $Y_k$, $L_k$ and $U_k$ are all
placed on the positive side, or in the negative side, and similarly for the negated symbols $-Z_k$, $-Y_k$, $-L_k$ and $-U_k$.
\begin{equation}\label{eq9}
\begin{array}{ll}
(Z_k, X, -Y_k) & (-Z_k, X, Y_k) \\
(Z_k, X, -L_k) & (-Z_k, X, L_k) \\
(Z_k, X, -U_k) & (-Z_k, X, U_k) \\
\end{array}
\end{equation}




First, if the $k^{th}$ clause is the first clause in the ordering that starts by $x_i$ or $-x_i$,
we need the following triples saying that $|x_i|<|U_k|$ and $|L_k|<|x_i|$.
\begin{equation}\label{eq11}
\begin{array}{ll}
(-U_k,x_i,U_k)& (-U_k, -x_i, U_k) \\
(-x_i, L_k, x_i)& (-x_i, -L_k, x_i) \\
\end{array}
\end{equation}

Second, if the $k^{th}$ clause is not the first one in the 
ordering that has $x_i$ or $-x_i$ as its first literal (i.e. $a_{k-1}, a_k\in\{-x_i,x_i\}$), we proceed by introducing the following triples. By induction, 
we assume that for the previous $(k-1)^{st}$ clause, the subintervals were created
using $L_{k-1}$
and $U_{k-1}$ with the corresponding $Z_{k-1}$ and $Y_{k-1}$ in them.
In this case, we need the triples saying that $|Z_{k-1}|<|U_k|$ and
$|Y_{k-1}|< |U_k|$ (instead of the first two triples of Equation 10):
\begin{equation}\label{eq11a}
\begin{array}{ll}
(-U_k, Z_{k-1}, U_k) & (-U_k, -Z_{k-1}, U_k) \\
(-U_k, Y_{k-1}, U_k) & (-U_k, -Y_{k-1}, U_k) \\
\end{array}
\end{equation}
and the triples saying that $|L_k|<|Z_{k-1}|$ and $|L_k|<|Y_{k-1}|$ (instead of the
last two triples of Equation 10).
\begin{equation}\label{eq12}
 \begin{array}{ll}
(-Z_{k-1}, L_k, Z_{k-1}) & (-Z_{k-1}, -L_k, Z_{k-1}) \\
(-Y_{k-1}, L_k, Y_{k-1}) & (-Y_{k-1}, -L_k, Y_{k-1})\\
\end{array}
\end{equation}

At this point, we have $Z_{k-1}$ either between two consecutive $L$ symbols,
or between two consecutive $U$ symbols. The same is true for $-Z_{k-1}$,
$Y_{k-1}$ and $-Y_{k-1}$.




Now, we need triples to locate $Z_k$, $-Z_k$, $Y_k$ and $-Y_k$ in the newly created intervals.
We add triples saying that $L_k$ and $U_k$ are between $Y_k$ and $Z_k$, and
similarly for the negated symbols.
\begin{equation}\label{eq13}
\begin{array}{ll}
(Z_k, U_k, Y_k) & (Z_k, L_k, Y_k) \\
(-Z_k, -U_k, -Y_k) & (-Z_k, -L_k, -Y_k) \\
\end{array}
\end{equation}
Since the $L$ and $U$ symbols are nested around $x_i$ and $-x_i$, this
puts the $Y_k$ and $Z_k$ symbols in their corresponding region around $x_i$ and $-x_i$.

Finally, we need to keep the $Z_k$, $-Z_k$, $Y_k$, and $-Y_k$ from being too far from $x_i$ or $-x_i$ and 
intruding in the region of $x_{i+1}$, $-x_{i+1}$, $x_{i-1}$, or $-x_{i-1}$.  This was done already for
the $Z_k$ and $-Z_k$ symbols in Equations \ref{eq6} and \ref{eq7}, and now we will do it for the $Y_k$ and $-Y_k$ symbols.
This is important in the case that the $k^{th}$
clause is the last one containing $x_i$ or $-x_i$  as its first literal. To do this, we need additional 
clauses saying that $|M_i|<|Y_k|$ and $|Y_k|<|M_{i+1}|$.

\begin{equation}\label{eq14}
\begin{array}{ll}
(Y_k, M_{i}, -Y_k) & (Y_k, -M_{i}, -Y_k) \\
(M_{i+1}, Y_{k}, -M_{i+1}) & (M_{i+1}, -Y_{k}, -M_{i+1}) \\
\end{array}
\end{equation}

\begin{theorem}\label{t:cBetweenness}
{\sc {\sc cBetweenness}} is ASP-complete.
\end{theorem}
\begin{proof}
Given a 3CNF instance of {\sc cNAE-3SAT} $\{a_k \vee b_k \vee c_k: 1\le k\le l\}$,
on the set of variables $\{x_1,\ldots,x_n\}$ and a subset of the constants $\{T,F\}$, we can create an instance of {\sc cBetweenness} 
by using Equations 1-14.
We show how to build a bijection between satisfying truth assignments for the formula instance
and betweenness orderings that are solution to the {\sc cBetweenness} instance.

To define the total betweenness orderings, 
we consider the line segment $[-n-1,n+1]$ and define a mapping $\phi_{\sigma}$ for 
every truth assignment $\sigma$ of the variables $x_1, \ldots, x_n$.  
While the domain of $\sigma$ is $\{x_1,\ldots,x_n\}$, the domain of $\phi_{\sigma}$, $S = dom(\phi_{\sigma})$, 
is contained in 
$$
\begin{array}{l}
\{x_1,\ldots,x_n, -x_1, \ldots, -x_n, 
m, M, X, 
M_1,\ldots,M_n, -M_1,\ldots,-M_n,\\
Z_1,\ldots, Z_l, -Z_1,\ldots, -Z_l, 
Y_1,\ldots, Y_l, -Y_1,\ldots, -Y_l,\\
L_1,\ldots, L_l, -L_1,\ldots, -L_l,
U_1,\ldots, U_l, -U_1,\ldots, -U_l
\}.
\end{array}
$$

The mapping $\phi_{\sigma}$ of truth assignments to orderings can now be defined the following way.
First as a general property of $\phi_{\sigma}$, let us say that $\phi_{\sigma}(-x)=-\phi_{\sigma}(x)$,
for every symbol $x$ of the instance.
If $\sigma(x_i)=T$, then $\phi_{\sigma}(x_i)=i$ (and $\phi_{\sigma}(-x_i)=-i$), 
and otherwise $\phi_{\sigma}(x_i)=-i$ (and $\phi_{\sigma}(-x_i)=i$). 
At this point the symbols, $x_i$ and $-x_i$ get fixed in the interval $[-n,n]$ on
opposite sides of $0$. Note that the symbol $X$ represents $0$ in the ordering, that is, $\phi_{\sigma}(X) = 0$.
This part of the definition of $\phi_{\sigma}$ fulfills Equation~\ref{eq1}.
Next, we put $m$ below $-n$, and $M$ above $n$ (i.e.~$\phi_{\sigma}(m) = -n-1$ and $\phi_{\sigma}(M) = n+1$), 
as is required for the definition of a betweenness ordering for {\sc cBetweenness}.

Next, we define the ordering function for the $M_i$ symbols. If $\phi_{\sigma}(x_i)>0$,
then $\phi_{\sigma}(M_i)=\phi_{\sigma}(x_i)-1/2$ and $\phi_{\sigma}(-M_i)=\phi_{\sigma}(-x_i)+1/2$. If $\phi_{\sigma}(x_i)<0$,
then $\phi_{\sigma}(M_i)=\phi_{\sigma}(x_i)+1/2$ and $\phi_{\sigma}(-M_i)=\phi_{\sigma}(-x_i)-1/2$. 
This definition fulfills Equations~\ref{eq2} and \ref{eq3}.

Now, for every $k,1\le k\le l$, we have to fix the position of every $Z_k$ under the mapping $\phi_{\sigma}$.
Recall that each clause $a_k\vee b_k\vee c_k$ is ordered from smaller to larger index. 
We begin with the case where only this clause begins with the variable represented by $a_k$.
Then, no additional auxiliary variables (i.e.~$Y_k$, $L_k$, and $U_k$) were introduced, and 
only $Z_k$ has to be placed in the order.  There are six cases based on possible truth values
assigned to $a_k$, $b_k$, and $c_k$ (as noted above, $a_k=b_k=c_k=false$ and $a_k=b_k=c_k=true$
are no satisfying truth assignments for an instance of {\sc cNAE-3SAT}. If
$\sigma(a_k)=\sigma(c_k)$ and $\sigma(b_k)$ has the opposite value (i.e.~the cases of $a_k=c_k=false$ and $b_k=true$,
and $a_k=c_k=true$ and $b_k=false$), 
then $Z_k$ will be set around $-\phi_{\sigma}(a_k)$, since the Equations~\ref{eq5} and \ref{eq6},
the fact that the index of $b_k$ is bigger than that of $a_k$,
and the fact that we have the triple $(Z_k, X, c_k)$ of Equation~\ref{eq3}. 
Note that the triples in Equation~\ref{eq7} fix the positions of $\phi_{\sigma}(Z_k)$ to one 
side of $-\phi_{\sigma}(a_k)$.
For the remaining cases where 
$\sigma(a_k)\neq\sigma(c_k)$ we define $\phi_{\sigma}$ to place $Z_k$ near $\phi_{\sigma}(a_k)$,
again by Equations~\ref{eq3}, \ref{eq5}, and \ref{eq6}. The general definition is as follows:

$$
\phi_{\sigma}(Z_k)=\left\{
\begin{array}{ll}
\phi_{\sigma}(a_k) + \frac{1}{4}     &\mbox{if $\sigma(a_k)\neq\sigma(c_k)$ and $\sigma(b_k)=true$}\\
\phi_{\sigma}(a_k) -\frac{1}{4}      &\mbox{if $\sigma(a_k)\neq\sigma(c_k)$ and $\sigma(b_k)=false$}\\
\phi_{\sigma}(-a_k) + \frac{1}{4}    &\mbox{if $\sigma(a_k)=\sigma(c_k)$ and $\sigma(b_k)=true$}\\
\phi_{\sigma}(-a_k)  -\frac{1}{4}    &\mbox{if $\sigma(a_k)=\sigma(c_k)$ and $\sigma(b_k)=false$}.
\end{array}
\right.
$$

If there is only one clause beginning with a given literal or its negation, then by the triples
in Equations~\ref{eq1} to \ref{eq9}, all the symbols (as well as 
their negations) are fixed.

There might be a number of clauses with the same variable in the first position
of the disjunction. The positions of these respective $Z_k$'s have to be fixed, as well
as the auxiliary variables, $L_k$, $U_k$ and $Y_k$'s.  By the triples in Equations \ref{eq12}, \ref{eq13}
and \ref{eq14} the auxiliary variables $L_k$ and $U_k$ (and the negative ones)
are nested around $\phi_{\sigma}(a_k)$ and $\phi_{\sigma}(-a_k)$ forming intervals, and $Z_k$ and $Y_k$
are set inside intervals of consecutive $L_k$'s or consecutive $U_k$'s (see Figure~\ref{fig:orderZ}).  
We assume that this is the $p^{th}$ clause that begins with the same variable.  Then,
as above, the exact placement of these variables depends on $\sigma$. Define $\phi_\sigma$ as follows:
\begin{itemize}
\itemsep 0pt
    \item{Case 1:} $\sigma(a_k) = false$, $\sigma(b_k) = false$, $\sigma(c_k) = true$:
        $$
         \begin{array}{ll}
            \phi_{\sigma}(L_k) = \phi_{\sigma}(a_k) + \frac{p}{2l}& 
            \phi_{\sigma}(U_k) = \phi_{\sigma}(a_k) - \frac{p}{2l}\\
            \phi_{\sigma}(Z_k) = \phi_{\sigma}(a_k) - \frac{2p+1}{4l}&
            \phi_{\sigma}(Y_k) = \phi_{\sigma}(a_k) + \frac{2p+1}{4l}\\
        \end{array}
        $$
    \item{Case 2:} $\sigma(a_k) = false$, $\sigma(b_k) = true$, $\sigma(c_k) = false$:
        $$
         \begin{array}{ll}
            \phi_{\sigma}(L_k) = -\phi_{\sigma}(a_k) - \frac{p}{2l}& 
            \phi_{\sigma}(U_k) = -\phi_{\sigma}(a_k) + \frac{p}{2l}\\
            \phi_{\sigma}(Z_k) = -\phi_{\sigma}(a_k) + \frac{2p+1}{4l}&
            \phi_{\sigma}(Y_k) = -\phi_{\sigma}(a_k) - \frac{2p+1}{4l}\\
        \end{array}
        $$
    \item{Case 3:} $\sigma(a_k) = false$, $\sigma(b_k) = true$, $\sigma(c_k) = true$:
        $$
         \begin{array}{ll}
            \phi_{\sigma}(L_k) = \phi_{\sigma}(a_k) + \frac{p}{2l}& 
            \phi_{\sigma}(U_k) = \phi_{\sigma}(a_k) - \frac{p}{2l}\\
            \phi_{\sigma}(Z_k) = \phi_{\sigma}(a_k) + \frac{2p+1}{4l}&
            \phi_{\sigma}(Y_k) = \phi_{\sigma}(a_k) - \frac{2p+1}{4l}\\
        \end{array}
        $$
    \item{Case 4:} $\sigma(a_k) = true$, $\sigma(b_k) = false$, $\sigma(c_k) = false$:
        $$
         \begin{array}{ll}
            \phi_{\sigma}(L_k) = \phi_{\sigma}(a_k) - \frac{p}{2l}& 
            \phi_{\sigma}(U_k) = \phi_{\sigma}(a_k) + \frac{p}{2l}\\
            \phi_{\sigma}(Z_k) = \phi_{\sigma}(a_k) - \frac{2p+1}{4l}&
            \phi_{\sigma}(Y_k) = \phi_{\sigma}(a_k) + \frac{2p+1}{4l}\\
        \end{array}
        $$
    \item{Case 5:} $\sigma(a_k) = true$, $\sigma(b_k) = false$, $\sigma(c_k) = true$:
        $$
         \begin{array}{ll}
            \phi_{\sigma}(L_k) = -\phi_{\sigma}(a_k) + \frac{p}{2l}& 
            \phi_{\sigma}(U_k) = -\phi_{\sigma}(a_k) - \frac{p}{2l}\\ 
            \phi_{\sigma}(Z_k) = -\phi_{\sigma}(a_k) - \frac{2p+1}{4l}&
            \phi_{\sigma}(Y_k) = -\phi_{\sigma}(a_k) + \frac{2p+1}{4l}\\
        \end{array}
        $$
    \item{Case 6:} $\sigma(a_k) = true$, $\sigma(b_k) = true$, $\sigma(c_k) = false$:
        $$
         \begin{array}{ll}
            \phi_{\sigma}(L_k) = \phi_{\sigma}(a_k) - \frac{p}{2l}& 
            \phi_{\sigma}(U_k) = \phi_{\sigma}(a_k) + \frac{p}{2l}\\ 
            \phi_{\sigma}(Z_k) = \phi_{\sigma}(a_k) + \frac{2p+1}{4l}&
            \phi_{\sigma}(Y_k) = \phi_{\sigma}(a_k) - \frac{2p+1}{4l}\\
        \end{array}
        $$
\end{itemize}
For each of these cases, it can be easily checked that the auxiliary variables satisfy 
Equations~\ref{eq9}-\ref{eq14}.

We have shown that given a satisfying truth assignment $\sigma$, there exists a mapping, 
$\phi_{\sigma}: S \rightarrow [-n-1,n+1]$.  This mapping induces a natural ordering on $S$ that
we will call $f_{\sigma}: S \rightarrow [0,|S|+1]$.  For $s_1,s_2\in S$, 
$$
    f_{\sigma}(s_1) < f_{\sigma}(s_2) \iff \phi_{\sigma}(s_1) < \phi_{\sigma}(s_2).
$$
Note that this fixes the minimal and maximal elements, $m$ and $M$ so that $f_{\sigma}(m) = 0$ and
$f_{\sigma}(M) = |S|+1$.  Further, we note that $f_{\sigma}$  satisfies Equations~\ref{eq1}-\ref{eq14}, 
since $\phi_{\sigma}$ satisfied them and has the same ordering.  We note that, by construction,
each satisfying truth assignment, $\sigma$, uniquely defines $f_{\sigma}$.  This construction
can be done in quadratic time in $|S|$.  Since the number of auxiliary variables in $S$ is bounded 
polynomially in $n$, we have a polynomial time reduction from {\sc cNAE-3SAT} to {\sc cBetweenness}.

Now we need to show the converse. Given a betweenness ordering on $S$ that satisfies all the triples obtained
from an instance, we need to define an assignment that for every clause, there is one literal
satisfied and one literal falsified. We define the assignment as follows. For all symbols $x_i$ (resp. $-x_i$) such that $f(x_i)>f(X)$ (resp. $f(-x_i)>f(X)$), we assign $x_i$ (resp. $-x_i$) to {\it true}. Similarly,  for all symbols $x_i$ (resp. $-x_i$) such that $f(x_i)<f(X)$ (resp. $f(-x_i)<f(X)$), we assign $x_i$ (resp. $-x_i$) to {\it false}. Furthermore, let $m$ be the constant $F$, and let $M$ be the constant $T$.
Now, we have to see that the assignment obtained
satisfies at least one literal, and falsifies at least one literal of every clause. Equation~\ref{eq5} ensures
that for a given clause $a_k\vee b_k\vee c_k$, not all three literals can be to the right of $X$ or to the left
of $X$. Therefore, the assignment created from the ordering will set at least one literal to {\it true} and at least
one literal to {\it false}.  Lastly, we note that if two betweenness orderings on $S$, $f_1$ and $f_2$, yield identical
truth assignments on $\{x_1,\ldots,x_n\}$, then, by Equations~\ref{eq1}-\ref{eq4}, $f_1$ and $f_2$ agree on the 
ordering of $\{x_1,\ldots,x_n,-x_1,\ldots,-x_n,m,M,X\}$.  Further, Equations~\ref{eq5}-\ref{eq14} fix the remainder
auxiliary variables, and as such, we must have $f_1 = f_2$.  
\end{proof}

\begin{figure}
\begin{center}
\includegraphics[height=0.95in]{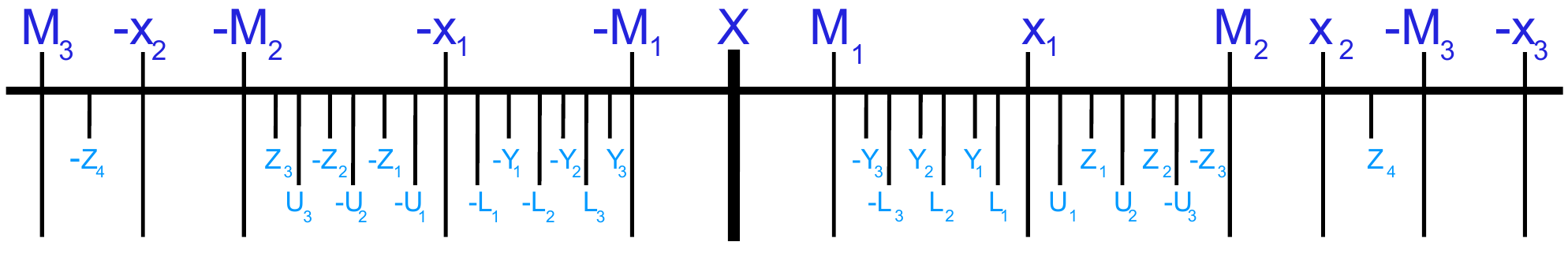}
\end{center}
\caption{\small Equations~\ref{eq9}-\ref{eq14} fix the locations of the $Z$ symbols, as well as the auxiliary variables
uniquely in the order.  Above illustrates the location of these auxiliary symbols for the $\sigma_1$ truth assignment
from Figure~\ref{fig:ex1}.}
\label{fig:orderZ}
\end{figure}

\subsection{The Quartet Challenge is coNP-complete}
In this section, we show that the {\sc Quartet Challenge} is coNP-complete. To this end, we extend the original argument
of Steel \cite[Theorem 1]{steel92} that showed that the related question, {\sc Quartet Compatibility},
is NP-complete.  He reduced {\sc Betweenness} to {\sc Quartet Compatibility} by mapping a betweenness
ordering to a caterpillar tree (for a definition, see below).  Under that reduction, multiple solutions to an instance $\psi$ of the
{\sc Quartet Compatibility} problem could correspond to a single solution of the {\sc Betweenness} instance that is obtained by transforming $\psi$.
To prove that the {\sc Quartet Challenge} is coNP-complete, we extend Steel's polynomial-time reduction from {\sc Betweenness} to an ASP-reduction from {\sc cBetweenness} to {\sc Quartet Compatibility}.

We first give some additional definitions. Let $A$ be a finite set, and let $C$ be a set of ordered triples of elements from $A$. Let $(x,y)$ be a pair of elements of $A$ such that no triple of $C$ contains $x$ and $y$. We say that $(x,y)$ is a {\it lost pair} with regards to $A$ and $C$.

Let $T$ be an unrooted phylogenetic tree. A pair of leaves $(a,b)$ of $T$ is called a {\it cherry} (or {\it sibling pair})
if $a$ and $b$ are leaves that are adjacent to a common vertex of $T$. Furthermore, $T$ is a {\it caterpillar} if $T$ is binary and has exactly two cherries. Following \cite{steel92}, we say that $T$ is an {\it $\alpha\beta$-caterpillar}, if $\alpha$ and $\beta$ are leaves of distinct cherries of $T$. We write $\alpha x_1|x_2x_3\ldots x_{n-1}|x_n\beta$ to denote the caterpillar whose two cherries are $(\alpha,x_1)$ and $(x_n,\beta)$, and, for each $i\in\{1,2,\ldots,n\}$, the path from $\alpha$ to $x_i$ consists of $i+1$ edges. Now, let $\alpha x_1|x_2x_3\ldots x_{n-1}|x_n\beta$ be an $\alpha\beta$-caterpillar. For each $i,j\in\{1,2,\ldots,n\}$ with $i\ne j$, we say that the path from $x_i$ to $x_j$ {\it crosses} $x_k$ if and only if $i<k<j$ or $j<k<i$. For example, Figure~\ref{fig:caterpillar} shows an $\alpha x_1|x_2x_3x_4|x_5\beta$-caterpillar whose path from $x_1$ to $x_4$ crosses $x_2$ and $x_3$.

Before we prove the main result of this section (Theorem~\ref{t:top2comp}), we need a lemma.

\begin{figure}
\center
\includegraphics{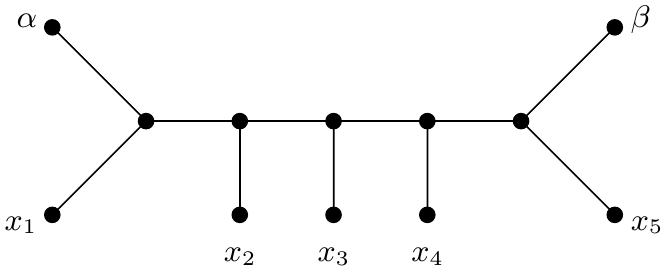}
\caption{\small The caterpillar $\alpha x_1|x_2x_3x_4|x_5\beta$.}
\label{fig:caterpillar}
\end{figure}

\begin{lemma}\label{l:cherry}
Let $T$ be a phylogenetic tree, and let $ab|cd$ be a quartet that is displayed by $T$. Then no element of $\{(a,c),(a,d),(b,c),(b,d)\}$ is a cherry of $T$.
\end{lemma}

\begin{proof}
By the definition of a quartet, the path from $a$ to $b$ in $T$ does not intersect the path from $c$ to $d$ in $T$. Thus, no element of $\{(a,c),(a,d),(b,c),(b,d)\}$ is a cherry of $T$.
\end{proof}

\begin{theorem}\label{t:top2comp}
{\sc Quartet Compatibility} is ASP-complete.
\end{theorem}

\begin{proof}
We start by noting that it clearly takes polynomial time to decide whether or not a phylogenetic tree $T$ displays a given set $Q$ of quartets since it is sufficient to check if $T|L(q)\cong q$ for each quartet $q$ in $Q$. Hence, {\sc Quartet Compatibility} is in FNP.

To show that {\sc Quartet Compatibility} is ASP-complete, we next describe an ASP-reduction from {\sc cBetweenness} to {\sc Quartet Compatibility}. 
Let $\psi$ be an instance of {\sc cBetweenness} over a finite set $A=\{a_1,a_2,\ldots,a_s\}\cup\{m,M\}$. Let $C=\{\pi_1,\pi_2,\ldots,\pi_n\}$ be the set of triples of $\psi$, with $\pi_i=(b_i,c_i,d_i)$ for each $i\in\{1,2,\ldots,n\}$, such that each element of $A\cup\{m,M\}$ is contained in at least one triple. Recall that $m$ is the first and $M$ the last element of each betweenness ordering $f$ of $A\cup\{m,M\}$ for $C$; that is $f(m)=0$ and $f(M)=|A|+1$. For simplicity throughout this proof, let $A'=A\cup\{m,M\}$. Furthermore, let $L=\{\tau_{n+1},\tau_{n+2},\ldots,\tau_{n'}\}$ precisely be the set of all lost pairs with regards to $A'$ and $C$, where $\tau_i=(x_i,y_i)$ for each $i\in\{n+1,n+2,\ldots,n'\}$. \\

We next describe six sets of quartets:
\begin{enumerate}
  \item Each triple $\pi_i=(b_i,c_i,d_i)$ in $C$ is represented by 6 quartets in $$Q_1=\bigcup_{i=1}^n Q_{\pi_i}= \bigcup_{i=1}^n\{p_ip_i'|b_ic_i, p_ib_i|c_id_i, p_ic_i|d_iq_i, p_id_i|q_iq_i',\alpha p_i|p_i'\beta,\alpha q_i|q_i'\beta\}.$$
  \item Each lost pair $\tau_i=(x_i,y_i)$ in $L$ is represented by 5 quartets in $$Q_2=\bigcup_{i=n+1}^{n'}Q_{\tau_i}=\bigcup_{i=n+1}^{n'}\{p_ip_i'|x_iy_i, p_ix_i|y_iq_i, p_iy_i|q_iq_i', \alpha p_i|p_i'\beta,\alpha q_i|q_i'\beta\}.$$
  \item Let $a_j,a_k\in A'$, and let $a$ be any fixed element of $A$. Set $Q_3$, $Q_4$, and $Q_5$ to be the following: $$Q_3=\bigcup_{i=2}^{n'}\bigcup_{j=1}^{i-1}\{p_ip'_i|p_ja, p_ip'_i|p_j'a, q_iq'_i|q_ja, q_iq'_i|q_j'a\},$$ $$Q_4=\bigcup_{i=1}^{n'}\bigcup_{j=1}^{n'}\{p_ip_i'|q_jq_j', p_ip_i'|q_ja, p_ip_i'|q_j'a\},\textnormal{ and}$$
 $$Q_5=\bigcup_{i=1}^{n'}\bigcup_{j=2}^{s+2}\bigcup_{k=1}^{j-1}\{p_ip_i'|a_ja_k, q_iq_i'|a_ja_k\}.$$
  \item Let $a_i,a_j\in A$, and set $Q_6$ to be the following: $$Q_6=\bigcup_{i=2}^{s}\bigcup_{j=1}^{i-1}\{\alpha m|a_ia_j, a_ia_j|M\beta\}.$$
\end{enumerate}
Noting that $n'$ is in the order of $O(|A|^3)$,  the quartet set $$Q=\bigcup_{i=1}^6 Q_i$$ can be constructed in polynomial time.\\

\noindent We note that for Steel's original proof~\cite[Theorem 1]{steel92}, in which he describes a polynomial-time reduction from {\sc Betweenness} to {\sc Quartet Compatibility} in order to show that the latter decision problem is NP-complete, the construction of $Q_1$ is sufficient.\\

A straightforward analysis of the quartets in $Q_{\pi_i}$ and $Q_{\tau_i}$, respectively, shows that ${<}Q_{\pi_i}{>}$ and ${<}Q_{\tau_i}{>}$ contain the following phylogenetic trees which are all $\alpha\beta$-caterpillars: $${<}Q_{\pi_i}{>}=\{\alpha p_i|p_i'b_ic_id_iq_i|q_i'\beta, \alpha q_i|q_i'd_ic_ib_ip_i|p_i'\beta\}$$ for each $i\in\{1,2,\ldots,n\}$ and $${<}Q_{\tau_i}{>}=\{\alpha p_i|p_i'x_iy_iq_i|q_i'\beta, \alpha q_i|q_i'y_ix_ip_i|p_i'\beta\}$$ for each $i\in\{n+1,n+2,\ldots,n'\}$. \\

Let $T$ be a phylogenetic tree of ${<}Q{>}$. By $Q_1$ and $Q_2$, it is easily checked that either, if $p_i$ and $p_i'$ are both crossed by the path from $\alpha$ to $b_i$ (resp. $x_i$) in $T$, then $q_i$ and $q_i'$ are both crossed by the path from $\beta$ to $b_i$ (resp. $x_i$) in $T$, or if $q_i$ and $q_i'$ are both crossed by the path from $\alpha$ to $b_i$ (resp. $x_i$) in $T$, then $p_i$ and $p_i'$ are both crossed by the path from $\beta$ to $b_i$ (resp. $x_i$) in $T$ for when $i\in\{1,2,\ldots,n\}$ (resp. $i\in\{n+1,n+2,\ldots,n'\}$). We refer to this property of $T$ as the {\it desired $pq$-property for $i$}.\\

Let $V$ be the set $\{p_1,\ldots,p_{n'},p_1',\ldots, p_{n'}',q_1,\ldots,q_{n'},q_1',\ldots,q_{n'}'\}$, and let $T$ be a phylogenetic tree in ${<}Q{>}$. We continue with making several observations that will be important in what follows:
\begin{enumerate}
\itemsep 0pt
    \item The second and third quartet in $Q_1$, the second quartet in $Q_2$, and Lemma~\ref{l:cherry} imply that
    $T$ does not have a cherry $(a,b)$ with $a,b\in A'$. 
    \item The quartets in $Q_5$ and Lemma~\ref{l:cherry} imply that $T$ does not have a cherry $(a,b)$ with $a\in A'$ and $b\in V$. 
    \item The last two quartets in $Q_1$ and $Q_2$, the quartets in $Q_3$, the first quartet in $Q_4$, and Lemma~\ref{l:cherry}
    imply that $T$ does not have a cherry $(a,b)$ with $a,b\in V$. 
\end{enumerate}
In conclusion, $T$ is an $\alpha\beta$-caterpillar. This observation leads to a number of additional properties that are satisfied by $T$:
\begin{itemize}
  \item [(i)] By $Q_5$ and the desired $pq$-property for each $i\in\{1,2,\ldots,n'\}$, the subtree $T(A')$ can be obtained from $T$ by deleting exactly two of its edges. 
  \item [(ii)] Let $i\in\{1,2,\ldots,n'\}$. Let $P$ contain each element of $(\{p_1,p_2,\ldots,p_{n'},p_1',p_2',\ldots,p_{n'}'\}-\{p_i,p_i'\})$ that is crossed by the path from $p_i$ (and $p_i'$) to $a$ in $T$ for any $a\in A'$. Then, by $Q_3$, each element in $P$ has an index that is smaller than $i$. Analogously, let $P'$ contain each element of $(\{q_1,q_2,\ldots,q_{n'},q_1',q_2',\ldots,q_{n'}'\}-\{q_i,q_i'\})$ that is crossed by the path from $q_i$ (and $q_i'$) to $a$ in $T$ for any $a\in A'$. Then, again by $Q_3$, each element in $P'$ has an index that is smaller than $i$.
  \item [(iii)] Let $a\in A'$. By the second and third quartet of $Q_4$ the path from $q_i$ (and $q_i'$) to $a$ does not cross an element of $\{p_1,p_2,\ldots,p_{n'},p_1',p_2',\ldots,p_{n'}'\}$ in $T$ for each $i\in\{1,2,\ldots,n'\}$.
  \item [(iv)] By (i)-(iii), the path from $p_i$ (resp. $q_i$) to $p_i'$ (resp. $q_i'$) in $T$ consists of 3 edges for each $i\in\{1,2,\ldots,n'\}$. In particular, by the last two quartets of $Q_1$ and $Q_2$, respectively, the path from $\alpha$ to $p_i'$ (resp. $q_i'$) crosses $p_i$ (resp. $q_i$) and the path from $\beta$ to $p_i$ (resp. $q_i$) crosses $p_i'$ (resp. $q_i'$) in $T$.
  \item[(v)] By $Q_6$, the path from $m$ to $M$ crosses each element in $A$. Furthermore, neither the path from $\alpha$ to $m$ nor the path from $\beta$ to $M$ crosses an element of $A$.
\end{itemize}

\noindent To illustrate, Figure~\ref{fig:top2comp} shows an $\alpha\beta$-caterpillar $T$ of ${<}Q{>}$ and, consequently, satisfies properties (i)-(v) for an instance of {\sc cBetweenness} that contains the four triples $\pi_1=(m,a,b)$, $\pi_2=(M,c,b)$, $\pi_3=(c,a,m)$, and $\pi_4=(m,b,M)$. Note that $T(A')$ can be obtained from $T$ by deleting the two edges $e$ and $e'$. \\

Now, let $T$ be a phylogenetic tree in ${<}Q{>}$. Let $T'$ be the phylogenetic tree obtained from $T$ by interchanging $\alpha$ and $\beta$, and for each $i\in\{1,2,\ldots,n'\}$, interchanging $p_i$ and $p_i'$, and $q_i$ and $q_i'$. Noting that $T'$ does display $Q\backslash Q_6$ but does not display $Q$ since property (v) is not satisfied, the rest of this proof essentially consists of two claims.\\

\begin{figure}
\center
\includegraphics[width=6in]{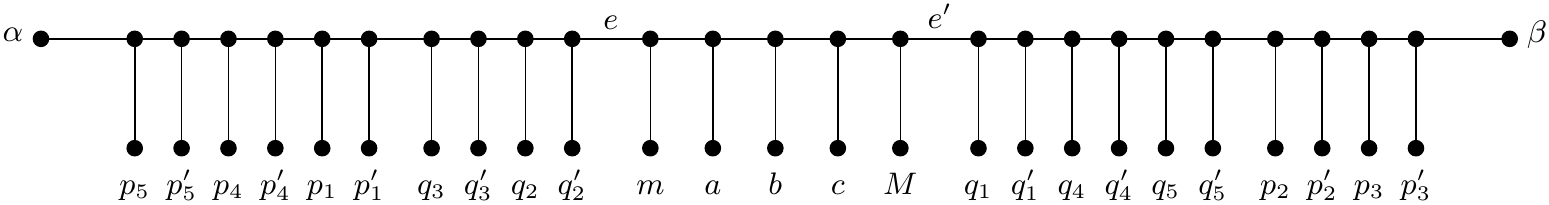}
\caption{\small An $\alpha\beta$-caterpillar in ${<}Q{>}$ that satisfies properties (i)-(v) for an instance of {\sc cBetweenness} that consists of the four triples $\pi_1=(m,a,b)$, $\pi_2=(M,c,b)$, $\pi_3=(c,a,m)$, and $\pi_4=(m,b,M)$. Note that the associated set of lost pairs only contains $\tau_5=(a,M)$. Details on how to construct $Q$ are given in in the proof of Theorem~\ref{t:top2comp}.}
\label{fig:top2comp}
\end{figure}

\noindent{\bf Claim 1.} Let $T$ and $T'$ be two elements of ${<}Q{>}$. Then $T\cong T'$ if and only if $T|(A'\cup\{\alpha,\beta\})\cong T'|(A'\cup\{\alpha,\beta\})$.\\

Trivially, if $T\cong T'$, then in particular $T|(A'\cup\{\alpha,\beta\})\cong T'|(A'\cup\{\alpha,\beta\})$. To prove the claim it is therefore sufficient to show that, if $T\ncong T'$, then $T|(A'\cup\{\alpha,\beta\})\ncong T'|(A'\cup\{\alpha,\beta\})$. Assume the contrary. Then there exist two distinct elements $T$ and $T'$ in ${<}Q{>}$ such that $T|(A'\cup\{\alpha,\beta\})\cong T'|(A'\cup\{\alpha,\beta\})$. Let $a$ be an element of $A'$. Since both of $T$ and $T'$ are $\alpha\beta$-caterpillars that satisfy properties (i)-(v), there exists an $i\in\{1,2,\ldots,n'\}$ such that the path from $\alpha$ to $a$ crosses $p_i$ and $p_i'$ in one of $T$ and $T'$, say $T$, while the path from $\alpha$ to $a$ crosses $q_i$ and $q_i'$ in $T'$. By the desired $pq$-property for each $i$ and property (i), note that the path from $\beta$ to $a$ crosses $q_i$ and $q_i'$ in $T$, and the path from $\beta$ to $a$ crosses $p_i$ and $p_i'$ in $T'$.
If $i\in\{1,2,\ldots,n\}$, let $S=\{\alpha,\beta,p_i,p_i',q_i,q_i',b_i,c_i,d_i\}$, and if $i\in\{n+1,n+2,\ldots,n'\}$, let $S=\{\alpha,\beta,p_i,p_i',q_i,q_i',x_i,y_i\}$.  Since $T|(S-\{p_i,p_i',q_i,q_i'\})\cong T'|(S-\{p_i,p_i',q_i,q_i'\})$, it now follows that either $T|S$ or $T'|S$ is not an element of ${<}Q_{\pi_i}{>}$ (if $i\in\{1,2,\ldots,n\}$) or ${<}Q_{\tau_i}{>}$ (if $i\in\{n+1,n+2,\ldots,n'\}$); thereby contradicting that $T$ and $T'$ are both in ${<}Q{>}$. This completes the proof of Claim 1.\\

\noindent{\bf Claim 2.} $Q$ is compatible if and only if $A'$ has a betweenness ordering $f$ for $C$ with $f(m)=0$ and $f(M)=|A|+1$. In particular, there is a bijection from the solutions of $\psi$ to the elements in ${<}Q{>}$.\\

First, suppose that $Q$ is compatible. Again, let $T$ be an unrooted binary phylogenetic tree in ${<}Q{>}$. Recall that the sets ${<}Q_{\pi_i}{>}$ and ${<}Q_{\tau_i}{>}$ both contain two $\alpha\beta$-caterpillars.
Thus $$T|\{p_i,p_i',q_i,q_i',\alpha,\beta,b_i,c_i,d_i\}$$ is isomorphic to one phylogenetic tree of ${<}Q_{\pi_i}{>}$ for each $i\in\{1,2,\ldots,n\}$, and $$T|\{p_i,p_i',q_i,q_i',\alpha,\beta,x_i,y_i\}$$ is isomorphic to one phylogenetic tree of ${<}Q_{\tau_i}{>}$ for each $i\in\{n+1,n+2,\ldots,n'\}$. Noting that $T$ is an $\alpha\beta$-caterpillar, we next define a betweenness ordering of $A'$ for $C$. Let $T^*$ be $T|\{A'\cup\{\alpha,\beta\}\}$, and define $f:A'\rightarrow\{1,2,\ldots,|A'|\}$ such that $2+f(a_j)$ denotes the number of edges on the path from $\alpha$ to $a_j$ in $T^*$ for each $a_j\in A'$. Since $T$ displays $Q_{\pi_i}$ for each $i\in\{1,2,\ldots,n\}$ and the path from $b_i$ to $d_i$ crosses $c_i$ in both phylogenetic trees of ${<}Q_{\pi_i}{>}$, it follows that $f(b_i)<f(c_i)<f(d_i)$ or $f(d_i)<f(c_i)<f(b_i)$. Since this holds for each $\pi_i\in C$, it follows that $f$ is a betweenness ordering of $A'$ for $C$. In particular, by property (v), we have $f(m)=0$ and $f(M)=|A|+1$. Furthermore, by Claim 1 and the paragraph prior to Claim 1, each element of ${<}Q{>}$ is mapped to a distinct betweenness ordering of $A'$ for $C$ with $f(m)=0$ and $f(M)=|A|+1$. \\

Second, suppose that $A'$ has a betweenness ordering for $C$, and let $f$ be one such ordering with $f(m)=0$ and $f(M)=|A|+1$. Note that $f$ imposes an ordering on each lost pair $(x_i,y_i)$ such that either $f(x_i)<f(y_i)$ or $f(y_i)<f(x_i)$. Furthermore, recall that $n'=|C|+|L|$ and $|A'|=s+2$. Let $T_0$ be the unique $\alpha\beta$-caterpillar, whose label set is $A'\cup\{\alpha,\beta\}$, such that the path from $\alpha$ to $a_j$ in $T_0$ contains $2+f(a_j)$ edges for each $a_j\in A'$. Let $a$ be any element of $A$. Next, we describe the algorithm {\sc BuildTree} that iteratively construct a series $T_1,T_2,\ldots,T_{n'}$ of $\alpha\beta$-caterpillars. Set $i$ to be 1. To obtain $T_i$ from $T_{i-1}$, proceed in the following way:\\

\begin{figure}
\center
\includegraphics{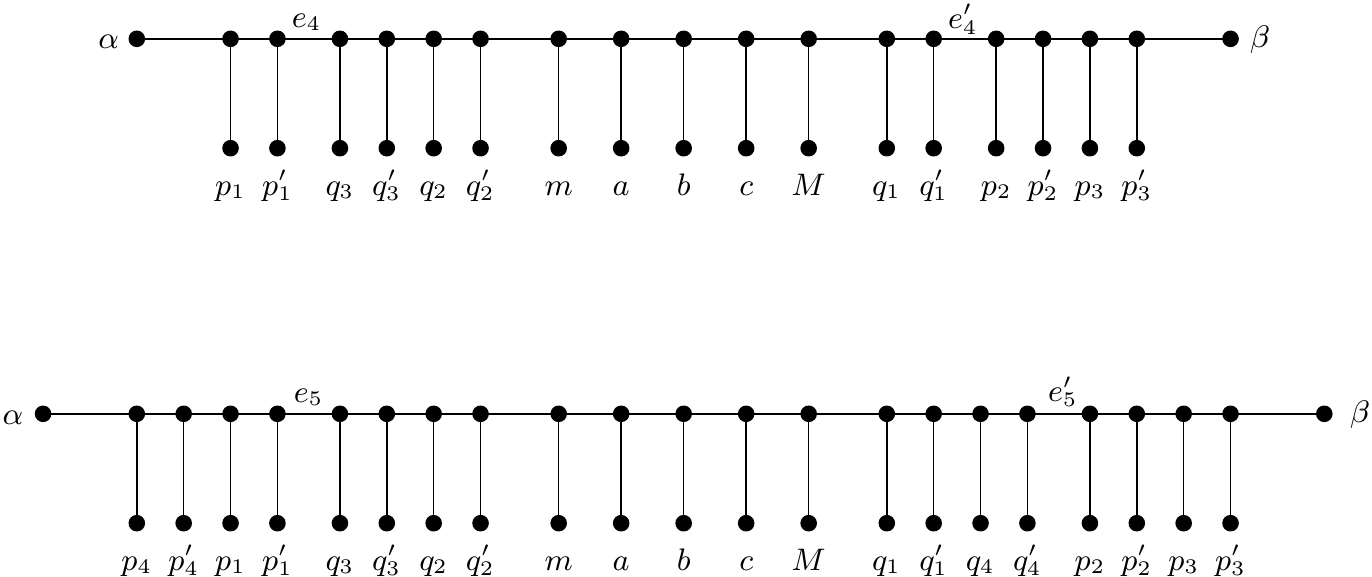}
\caption{\small The intermediate trees $T_3$ (top) and $T_4$ (bottom) that are obtained from applying the algorithm {\sc BuildTree} to the {\sc Betweenness} instance that is described in the caption of Figure~\ref{fig:top2comp} for when $m<a<b<c<M$ is the given betweenness ordering. Note that $T_4$ is obtained from $T_3$ by subdividing twice the edge that is incident with $\alpha$ and $e_4'$, respectively. Furthermore, the tree depicted in Figure~\ref{fig:top2comp} is obtained from $T_4$ by applying one more iteration of {\sc BuildTree}.}
\label{fig:ordering2tree}
\end{figure}

Let $P=\{p_1,p_2,\ldots,p_{i-1},p_1',p_2',\ldots,p_{i-1}'\}$.  We first define two edges $e_i$ and $e_i'$ in $T_{i-1}$. If the path from $\alpha$ to $a$ in $T_{i-1}$ crosses an element of $P$, let $e_i=\{u,v\}$ be the first edge on this path such that $u$ is adjacent to a leaf labeled with an element of $P$ and $v$ is adjacent to a leaf labeled with an element that is not contained in $P$. 
Otherwise, choose $e_i$ to be the edge that is incident with $\alpha$. Similarly, if the path from $\beta$ to $a$ in $T_{i-1}$ crosses an element of $P$, let $e_i'=\{u',v'\}$ be the first edge on this path such that $u'$ is adjacent to a leaf labeled with an element of $P$ and $v'$ is adjacent to a leaf labeled with an element that is not contained in $P$. Otherwise, choose $e_i'$ to be the edge that is incident with $\beta$. Note that $e_i$ and $e_i'$ are uniquely defined. 

To obtain $T_i$ from $T_{i-1}$, we consider two cases. First, if $i\in\{1,2,\ldots,n\}$ and $f(b_i)<f(c_i)<f(d_i)$, or if $i\in\{n+1,n+2,\ldots,n'\}$ and $f(x_i)<f(y_i)$, subdivide the edge incident with $\alpha$ twice and join each of the two newly created vertices with a new leaf labeled $p_i$ and $p_i'$, respectively, by introducing two new edges such that the path from $\alpha$ to $p_i'$ crosses $p_i$. Furthermore, subdivide $e_i'$ twice and join each of the two newly created vertices with a new leaf labeled $q_i'$ and $q_i$, respectively, by introducing two new edges such that the path from $\beta$ to $q_i$ crosses $q_i'$. Second, if $i\in\{1,2,\ldots,n\}$ and $f(d_i)<f(c_i)<f(b_i)$, or if $i\in\{n+1,n+2,\ldots,n'\}$ and $f(y_i)<f(x_i)$, subdivide $e_i$ twice and join each of the two newly created vertices with a new leaf labeled $q_i$ and $q_i'$, respectively, by introducing two new edges such that the path from $\alpha$ to $q_i'$ crosses $q_i$. Furthermore, subdivide the edge incident with $\beta$ twice and join each of the two newly created vertices with a new leaf labeled $p_i'$ and $p_i$, respectively, by introducing two new edges such that the path from $\beta$ to $p_i$ crosses $p_i'$. 
A specific example of the definition of $e_i$ and $e_i'$, respectively, and on how to obtain $T_i$ from $T_{i-1}$ is shown in Figure~\ref{fig:ordering2tree}.

If $i<n'$, increment $i$ by 1 and repeat; otherwise stop. In this way, we obtain a tree $T_{n'}$ that displays $Q$ and, hence $Q$ is compatible. In particular, $T_{n'}$ is an element of ${<}Q{>}$. Furthermore, again by Claim 1 and the paragraph prior to Claim 1, $T_{n'}$ is the unique tree of ${<}Q{>}$ that has the property that $T_{n'}|(A'\cup\{\alpha,\beta\})\cong T_0$. Thus, each betweenness ordering $f$ of $A'$ for $C$ with $f(m)=0$ and $f(M)=|A|+1$ is mapped to a distinct element of ${<}Q{>}$. This completes the proof of Claim 2.\\


It now follows that the presented transformation from an instance of {\sc cBetweenness} to an instance of {\sc Quartet Compatibility} is an ASP-reduction that can be carried out in polynomial time. Hence, {\sc Quartet Compatibility} is ASP-complete. This establishes the proof of this theorem.
\end{proof}

Now recall that ASP-completeness implies NP-completeness of the corresponding decision problem, say $\Pi_d$~\cite[Theorem 3.4] {yato03}. Since $\Pi_d$ is exactly the complementary question of the {\sc Quartet Challenge} (see last paragraph of Section~\ref{sec:prelim}), the next corollary immediately follows.

\begin{corollary}\label{cor:np}
The {\sc Quartet Challenge} is coNP-complete.
\end{corollary}

\section{Conclusion}
\label{sec:discussion}
In this paper, we have shown that the two problems {\sc cBetweenness} and {\sc Quartet Compatibility} that have applications in computational biology  are ASP-complete. Thus, given a betweenness ordering or a phylogenetic tree that displays a set of quartets, it is computationally hard to decide if another solution exists to a problem instance of {\sc cBetweenness} and {\sc Quartet Compatibility}, respectively. If there is another solution, then this may imply that a data set that underlies an analysis does not contain enough information to obtain an unambiguous result. Furthermore, by Corollary~\ref{cor:np}, the ASP-completeness of {\sc Quartet Compatibility} implies that the {\sc Quartet Challenge}, which is one of Mike Steel's \$100 challenges~\cite{quartetChallenge}, is coNP-complete.  Lastly, due to~\cite[Theorem 3.4]{yato03}, regardless of how many solutions to an instance of {\sc cBetweenness} or {\sc Quartet Compatibility} are known, it is always NP-complete to decide whether an additional solution exists.

Unless P=NP, the existence of efficient algorithms to exactly solve the above-mentioned two problems is unlikely. Nevertheless, there is a need to develop exact algorithms that solve small to medium sized problem instances and, most importantly, return all solutions. For example, it might be possible to start filling this gap by using fixed-parameter algorithms that have recently proven to be particularly useful to approach many questions in computational biology~\cite{gramm07}.
Alternatively, heuristics and polynomial-time approximation algorithms often provide a valuable tool to efficiently approach problem instances of larger size. While Chor and Sudan~\cite{chor98} established a geometric approach to approximate a betweenness ordering that satisfies at least one half of a given set of constraints, the statement of {\sc Quartet Compatibility} does not directly allow for an approximation algorithm since it is a recognition-type problem. Nevertheless, since compatible quartet sets are rare, the goal of a related problem is, given a set of quartets, to reconstruct a phylogenetic tree that displays as many quartets as possible. This problem is known as the {\sc Maximum Quartet Consistency} problem. Despite its NP-hardness~\cite{berry99,steel92}, 
several exact algorithms (e.g. see~\cite{ben-dor98,wu05} and references therein) as well as a polynomial-time approximation~\cite{jiang00} exist. It would therefore be interesting to investigate if these algorithms can be extended in a way such that they return all solutions in order to analyze whether a unique phylogenetic tree displays a given set of compatible quartets.

We end this paper by noting that the computational complexity of the {\sc Quartet Challenge} changes greatly if all elements in a set of quartets over $n$ taxa have a common taxa, say $x$. By rooting each quartet at $x$, i.e. deleting the vertex labeled $x$ and its incident edge and regarding the resulting degree-2 vertex as the root, we obtain a set $S$ of rooted triples (rooted phylogenetic trees on three taxa). By applying the {\sc Build} algorithm it can now be checked in polynomial time if $S$ is compatible~\cite[Proposition 6.4.4]{sempleSteelBook}. Furthermore, there is a unique rooted phylogenetic tree on $n$ taxa that displays $S$ if and only if {\sc Build} returns a rooted binary phylogenetic tree~\cite[Proposition 2]{bryantSteel95}. If {\sc Build} returns a rooted phylogenetic tree that is not binary, then every refinement of this tree also displays $S$.

\section{Acknowledgments }
We would like to thank Juraj Stacho (University of Haifa) for pointing out an oversight in an earlier version of this paper, and Ward Wheeler (American Museum of Natural History) for insightful conversations.  
The project was partially supported by grants from the Spanish Government (TIN2007-68005-C04-03 to Bonet and Linz), (TIN2010-20967-C04-04 to Bonet)
and from the 
US National Science Foundation (\#09-20920 to St.~John).

\small
\bibliographystyle{IEEEtranS}
\bibliography{quartet}

\end{document}